\documentclass[10pt,conference]{IEEEtran}
\IEEEoverridecommandlockouts
\usepackage[dvipdf]{graphicx,color}

\usepackage{amssymb}
\usepackage{amsmath}
\usepackage{stfloats}
\usepackage{amsfonts}
\usepackage{balance}
\usepackage{color}
\usepackage{algorithm}
\usepackage{algpseudocode}
\usepackage{color}
\usepackage{varwidth}
\usepackage{multicol}
\usepackage{subfigure}
\usepackage{xspace}
\usepackage{enumerate}
\usepackage{xcolor,cite,etoolbox}
\usepackage{bm}
\usepackage{amsthm}

\newtheorem{lemma}{Lemma}
\newtheorem{proposition}{Proposition}

\newtheorem{remark}{Remark}
\usepackage[font=footnotesize,labelfont=footnotesize]{caption}
\usepackage{ragged2e}

\def\BibTeX{{\rm B\kern-.05em{\sc i\kern-.025em b}\kern-.08em
    T\kern-.1667em\lower.7ex\hbox{E}\kern-.125emX}}

\begin{document}

\title{Performance Analysis of HAPS-Assisted Downlink RSMA under Orthogonal and Full Frequency Reuse}

\author{\IEEEauthorblockN{Farjam Karim\IEEEauthorrefmark{1},   Prathapasinghe Dharmawansa\IEEEauthorrefmark{1}, Nurul Huda Mahmood\IEEEauthorrefmark{1}, and Matti Latva-aho\IEEEauthorrefmark{1}%
	\thanks{This research was supported by Interreg Aurora ENSURE-6G Project. }}\\
\IEEEauthorblockA{\IEEEauthorrefmark{1}Centre for Wireless Communications, University of Oulu, Finland. \\  
		Email: \{farjam.karim,\;Prathapasinghe.KaluwaDevage,\;nurulhuda.mahmood,\;matti.latva-aho\}@oulu.fi 
		}} 

\maketitle

\begin{abstract}
This paper analyzes the outage performance of a high-altitude platform station (HAPS)-assisted downlink employing rate-splitting multiple access (RSMA). A realistic link budget accounting for free-space path loss, rain attenuation, and atmospheric absorption is combined with elevation-angle-dependent shadowed Rician fading. Closed-form outage probability and throughput expressions are derived for orthogonal frequency allocation and full frequency reuse, where the aggregate inter-beam interference is approximated by a moment-matched Gamma random variable, yielding a tractable finite-sum expression via its Laplace transform. The analytical expressions are validated through Monte Carlo simulations, showing close agreement. Numerical results show that the optimal common stream power allocation depends on the transmit power, while inter-beam interference under frequency reuse introduces an outage floor absent under orthogonal allocation, providing practical design insights for HAPS-RSMA systems.
\end{abstract}


\section{Introduction}
\label{sec:introduction}

Non-terrestrial networks (NTNs) are widely regarded as an integral part of the sixth-generation (6G) wireless vision, extending connectivity beyond the reach of terrestrial infrastructure and providing resilience where such infrastructure is unavailable or disrupted~\cite{Faical_mag_2025}. Among the platforms comprising the NTN ecosystem; geostationary and low-earth-orbit satellites, unmanned aerial vehicles, and high-altitude platform stations (HAPSs); HAPSs occupy a distinctive middle tier. Operating in the stratosphere at altitudes of $17$-$25$~kilometer (km), provide wide-area coverage with a predominantly line-of-sight (LoS) channel to ground users, combined with a substantially shorter propagation path, and hence lower latency, than either low-earth-orbit or geostationary satellites~\cite{kanani_2025_twc, kiric_ojcom_2025}.

A single HAPS beam can typically serve a large number of ground users over a limited spectrum allocation, making the choice of multiple access scheme central to how efficiently that beam is used. Under orthogonal multiple access (OMA), each user is allocated a disjoint share of the available resources, so the rate available to any one user falls directly as more users share the beam~\cite{Ghosh_comaprision_NOMA_OMA_Access_Jan23}. Non-orthogonal multiple access (NOMA) exploits this disparity by superposing users in the power domain, but requires every user to fully decode and cancel the signals of users with weaker channels, making its performance highly sensitive to the accuracy of this ordering~\cite{Bruno_proc}. Interestingly, rate-splitting multiple access (RSMA) is capable of relaxing this rigidity by splitting each user's message into a common message and a private message. The common messages of all users are jointly encoded into a single common stream, which is decoded by all users, while the private messages are individually encoded into separate private streams, each decoded only by its intended recipient~\cite{Bruno_proc, Farjam_twc_2025}. Consequently, RSMA allows interference to be partially decoded and partially treated as noise, thereby providing greater robustness and flexibility than power-domain NOMA~\cite{Bruno_proc, Farjam_twc_2025}.

Despite its advantages, the outage performance of downlink RSMA in HAPS-assisted multibeam systems remains largely unexplored (see, e.g., Section I of~\cite{huang_tvt_2025} and references therein).  HAPS deployments are expected to employ multiple beams to extend coverage and increase system capacity~\cite{Jved_twc_2025}. These beams may either operate on orthogonal frequency or time resources, thereby eliminating inter-beam interference~\cite{Jved_twc_2025}, or reuse the same resources to maximize spectral efficiency at the expense of additional interference. Although RSMA is inherently well suited to managing interference, the impact of these two operating regimes on the outage performance of HAPS systems has not been characterized analytically. Moreover, realistic HAPS links experience propagation conditions that differ from terrestrial networks, requiring the joint consideration of free-space path loss, atmospheric absorption, rain attenuation, and elevation-dependent shadowed-Rician fading. These observations motivate a unified analytical framework capable of characterizing the outage performance of HAPS-assisted RSMA under both interference-free and interference-limited multibeam operation.

The main contributions of this paper are threefold. First, we derive closed-form outage probability expressions for both the common and private streams of downlink HAPS-RSMA under orthogonal frequency allocation across beams, considering a realistic HAPS link budget and elevation-angle-dependent shadowed-Rician fading. Second, for the more challenging case of full frequency reuse, we develop a tractable analytical framework by approximating the aggregate inter-beam interference with a moment-matched Gamma random variable, leading to an approximated closed-form outage probability expressions. Finally, the developed analytical framework is validated through Monte Carlo simulations and used to quantify the impact of transmit power, RSMA power allocation, carrier frequency and multibeam frequency reuse on the outage and throughput performance of HAPS-RSMA. 

\noindent \textbf{Notations:}
$\mathbb{E}[\cdot]$ and $|\cdot|$ denote expectation and modulus operator, respectively. $\mathcal{CN}(0,N_0)$ denotes a circularly symmetric complex Gaussian distribution with zero mean and variance $N_0$. $f_X(\cdot)$, $F_X(\cdot)$ denote the PDF and CDF of $X$. ${}_1F_1(\cdot;\cdot;\cdot)$ is the confluent hypergeometric function, $\binom{\cdot}{\cdot}$ the binomial coefficient, and $(x)_p\triangleq x(x+1)\cdots(x+p-1)$ the rising factorial, with $(x)_0\triangleq1$. Superscript $(k)$ and subscript $n$ index the beam and user, respectively, and are omitted where unambiguous.
\section{System Model}
\label{sec:system_model}

\subsection{Network and Geometric Model}

We consider the downlink of a HAPS communication system, where a HAPS operating at a fixed altitude $H$ above the ground provides wireless connectivity to ground users as illustrated in Fig.~\ref{system_mod_fig}. To extend coverage, the HAPS employs multi-beam transmission and generates $\mathcal{K}\in\{1,2,\ldots,K\}$ beams. Within each beam, RSMA is adopted to simultaneously serve $\mathcal{N}\in\{1,2,\ldots,N\}$ users. Since all beams follow an identical transmission framework, the subsequent analysis is presented for a representative beam. Two scenarios are considered: (i) orthogonal frequency allocation, where inter-beam interference is absent, and (ii) full frequency reuse, where inter-beam interference is taken into account.

The position of $n_{th}$ user ($U_n$) relative to the HAPS $k_{th}$-beam is characterized by its elevation angle $\theta^{(k)}_n$, defined as the angle between the local horizontal plane at the $n_{th}$ user location and the line-of-sight direction toward the HAPS $k_{th}$-beam. 
Given the platform altitude $H$ and the elevation angle $\theta^{(k)}_n$ of user $U^{(k)}_n$, the slant propagation distance between the HAPS and $U^{(k)}_n$ is obtained from the right-triangle relationship between altitude, elevation angle, and slant range as
$d^{(k)}_n = \frac{H}{\sin\theta^{(k)}_n}$.
\begin{figure}[t!]
    \centering
\includegraphics[width=0.95\linewidth]{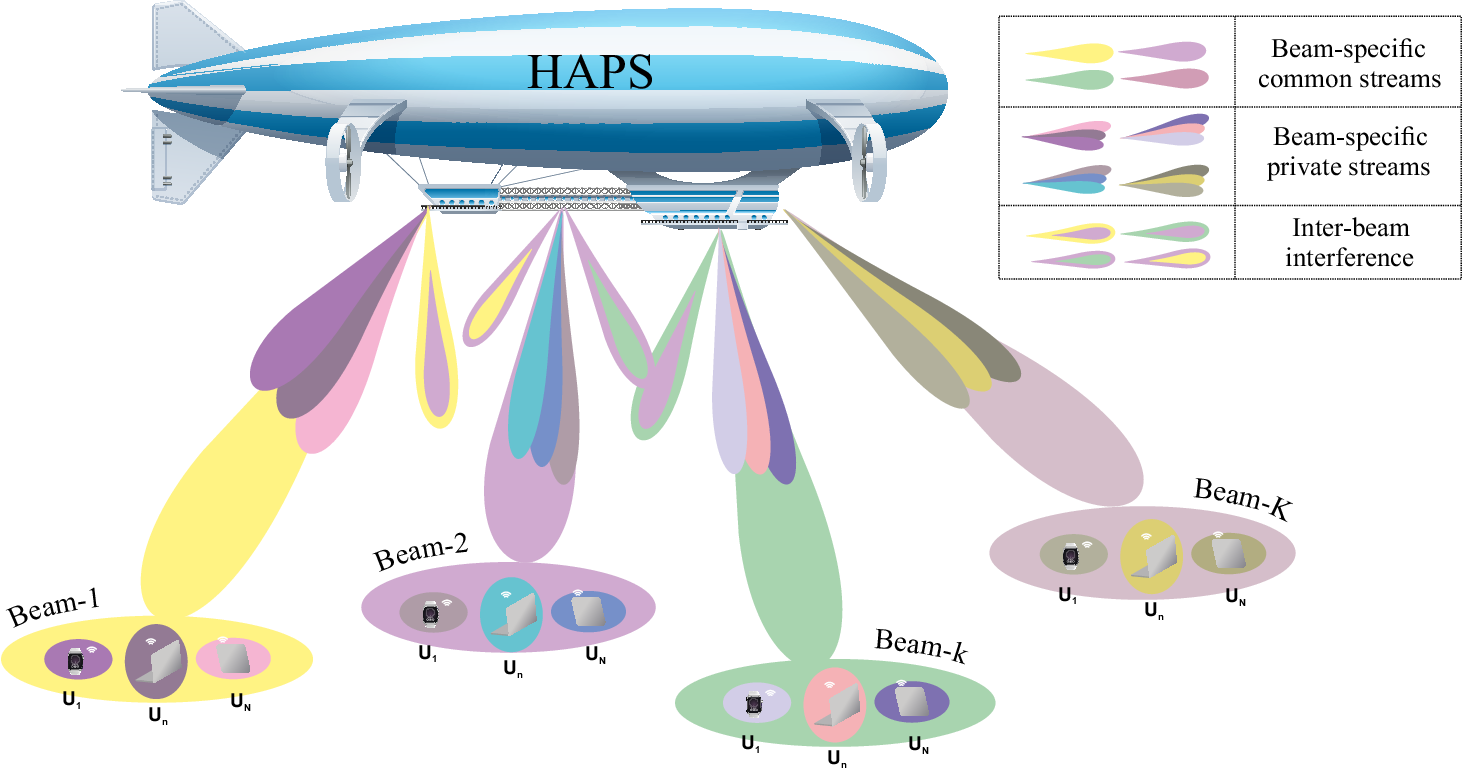}
    \caption{A schematic of the system model.}
    \label{system_mod_fig}
    \vspace{-1em}
\end{figure}
The large-scale channel gain between the HAPS and $U^{(k)}_n$ is characterized in the decibel (dB) domain as follows~\cite{Yahia_aerospace_2022}:
\begin{align}
F^{(k)}_n = G^{(k)}_T + G^{(k)}_{R,n} - L^{(k)}_{F,n} - L^{(k)}_{R,n} - L^{(k)}_A,
\label{eq:link_budget}
\end{align}
where, $G^{(k)}_T$ denotes the transmit antenna gain of the HAPS $k_{th}$-beam payload, and $G^{(k)}_{R,n}$ denotes the receive antenna gain of $U^{(k)}_n$, both expressed in dBi. The term $L^{(k)}_{F,n}$ represents the free-space path loss experienced by user $U^{(k)}_n$, and is given by
$L^{(k)}_{F,n} \ [\mathrm{dB}] = 92.45 + 20\log_{10}\!\big(f_r[\mathrm{GHz}]\big) + 20\log_{10}\!\big(d^{(k)}_n[\mathrm{km}]\big),$
where $f_r$ is the carrier frequency expressed in gigahertz (GHz) and $d^{(k)}_n$ is the slant propagation distance to user $U^{(k)}_n$, expressed in kilometers. The term $L^{(k)}_{R,n}$ accounts for attenuation due to rain along the propagation path, expressed in dB, and is modeled as the product of a rain attenuation rate (in dB/km) and the slant path length traversed. The term $L^{(k)}_A$ represents the fixed attenuation due to absorption by atmospheric gases (principally oxygen and water vapor), expressed in dB.

The quantity $F^{(k)}_n$ therefore represents the net large-scale link gain for user $U^{(k)}_n$, combining the beneficial effect of the antenna gains with the detrimental effects of free-space spreading loss, rain, and atmospheric absorption. Converting this quantity from the decibel domain to the linear scale yields the corresponding deterministic large-scale propagation gain
$L^{(k)}_n = 10^{F^{(k)}_n/10}$,
which will be used in the composite channel model below. 

\subsection{Small-Scale Fading: shadowed-Rician Model}

In addition to the deterministic large-scale propagation gain described above, the HAPS-to-ground link is subject to random small-scale fading. Owing to the dominant line-of-sight (LoS) component and its susceptibility to random shadowing, the small-scale fading of each link is modeled using the shadowed-Rician distribution~\cite{Yahia_aerospace_2022, yun_ai_photon_2019}.
Let $g^{(k)}_n$ denote the small-scale complex fading coefficient of the link to user $U^{(k)}_n$, and let
$Z^{(k)}_n \triangleq |g^{(k)}_n|^2$
denote the corresponding channel power gain. The PDF of $Z^{(k)}_n$ is given by~\cite{yun_ai_photon_2019}:
\begin{align}
f_{Z^{(k)}_n}(z)
=
\alpha^{(k)}_n
\exp(-\beta^{(k)}_n z)
\,{}_1F_1(m^{(k)}_n;1;\delta^{(k)}_n z), z\ge0,
\label{eq:sr_pdf}
\end{align}
where ${}_1F_1(\cdot;\cdot;\cdot)$ denotes the confluent hypergeometric function of the first kind. The parameter $m^{(k)}_n$ is a positive integer that characterizes the severity of shadowing affecting the LoS component. Specifically, smaller values of $m^{(k)}_n$ correspond to severe shadowing, whereas larger values indicate milder shadowing, approaching the conventional Rician fading model as $m^{(k)}_n\rightarrow\infty$. The parameters $\alpha^{(k)}_n$, $\beta^{(k)}_n$, and $\delta^{(k)}_n$ depend on the average LoS power $\Omega^{(k)}_n$ and the average scattered-component power $2b^{(k)}_n$, and are expressed as
$\alpha^{(k)}_n=\frac{1}{2b^{(k)}_n}\left(
\frac{2b^{(k)}_n m^{(k)}_n}{2b^{(k)}_n m^{(k)}_n+\Omega^{(k)}_n}
\right)^{m^{(k)}_n},
\beta^{(k)}_n=\frac{1}{2b^{(k)}_n},
\delta^{(k)}_n=\frac{\Omega^{(k)}_n}
{2b^{(k)}_n(2b^{(k)}_n m^{(k)}_n+\Omega^{(k)}_n)}$.

Since the shadowing parameter $m^{(k)}_n$ is assumed to be a positive integer, Kummer's transformation is first applied where the first argument becomes the non-positive integer. Consequently, the confluent hypergeometric function terminates into a finite polynomial, yielding a finite-series representation of the PDF expressed as
\begin{align}
f_{Z^{(k)}_n}(z)
=
\alpha^{(k)}_n
e^{-(\beta^{(k)}_n-\delta^{(k)}_n)z}
\sum_{i=0}^{m^{(k)}_n-1}
\binom{m^{(k)}_n-1}{i}\nonumber\\
\times
\frac{(\delta^{(k)}_n)^i}{i!}
z^i,
 z\ge0.
\label{eq:sr_pdf_series}
\end{align}

Integrating \eqref{eq:sr_pdf_series}, the CDF is obtained as
\begin{align}
F_{Z^{(k)}_n}(z)=1-&\alpha^{(k)}_n\sum_{i=0}^{m^{(k)}_n-1}\binom{m^{(k)}_n-1}{i}\frac{(\delta^{(k)}_n)^i}{i!}e^{-(\beta^{(k)}_n-\delta^{(k)}_n)z}\nonumber\\
&\times\sum_{l=0}^{i}\frac{i!}{l!}\frac{z^l}{(\beta^{(k)}_n-\delta^{(k)}_n)^{\,i-l+1}}, \qquad
 z\ge0.
\label{eq:sr_cdf}
\end{align} 
\subsection{Transmit Signal and SINR Characterization}

Within the  $k_{th}$-beam under consideration, the HAPS applies  RSMA to serve $\mathcal{N}$ users simultaneously. At the transmitter side each user's message is split by the rate-splitting encoder into two parts: a common part and a private part~\cite{Bruno_proc, Farjam_twc_2025}. The common parts of all users are combined and encoded into a single common stream $s^{(k)}_c$, drawn from a shared codebook that all the users within $k_{th}$-beam are capable of decoding, while each user's private part is encoded into an individual private stream, denoted $s^{(k)}_n$ for the $n_{th}$ user in the $k_{th}$-beam. These streams are then linearly superposed in the power domain to form the transmitted signal:
$x^{(k)} = \sqrt{P^{(k)} \zeta^{(k)}_c}\,s^{(k)}_c + \sum\limits^{N}_{n=1}\sqrt{P^{(k)} \zeta^{(k)}_n} s^{(k)}_n$,
where $\mathbb{E}[|s^{(k)}_c|^2]=\mathbb{E}[|s^{(k)}_n|^2]=1$, and $P^{(k)}$ denotes the power budget available for transmitting the $k_{th}$-beam, whereas $\zeta^{(k)}_c$ and $\zeta^{(k)}_n$ denotes the power allocation for common stream and the $U^{(k)}_n$ private stream, respectively. Note that $\zeta^{(k)}_c + \sum\limits^{N}_{n=1}\zeta^{(k)}_n=1$.
\subsubsection{Case I: Orthogonal Frequency Allocation}
In the first scenario, orthogonal frequency allocation is adopted among the HAPS beams. Therefore, each beam operates on an exclusive frequency resource, and inter-beam interference is completely avoided. The received signal at user $U_n^{(k)}$ ($k_{th}$ beam, $n_{th}$ user) is given by
\begin{align}
y^{(k)}_n = g^{(k)}_n \sqrt{L^{(k)}_n} x^{(k)}+  \eta^{(k)}_n
\label{eq:rx_signal}
\end{align}
where $\eta^{(k)}_n \sim \mathcal{CN}(0, N_0)$ is the additive white Gaussian noise (AWGN).

Decoding at each user proceeds in two stages, consistent with the RSMA decoding principle. In the first stage, user $U^{(k)}_n$ attempts to decode the common stream $s^{(k)}_c$, treating all the private streams as additional interference. From \eqref{eq:rx_signal}, the resulting signal-to-interference-plus-noise ratio (SINR) experienced by user $U^{(k)}_n$ when decoding the common stream is
\begin{equation}
\gamma^{(k)}_{c,n} = \frac{P^{(k)}\zeta^{(k)}_c|g^{(k)}_n|^2L^{(k)}_n}{P^{(k)}\left(1-\zeta^{(k)}_c\right)|g^{(k)}_n|^2L^{(k)}_n + N_0}.
\label{eq:sinr_common}
\end{equation}
 Because the common stream must be reliably decodable by every user sharing the beam before SIC can proceed, the maximum achievable rate for the common stream is limited by whichever user experiences the lowest of all the common stream SINRs.
 

Once user $U^{(k)}_n$ has successfully decoded the common stream, it reconstructs the corresponding signal component and subtracts it from the received signal $y^{(k)}_n$. Following this cancellation, $U^{(k)}_n$ proceeds to the second decoding stage, in which it decodes its own private stream $s^{(k)}_n$, while treating the other private streams $s^{(k)}_i$ (where $i \neq n$)  as  interference. The resulting SINR can be expressed as
\begin{align}
\gamma^{(k)}_{p,n} = \frac{P^{(k)}\zeta^{(k)}_n|g^{(k)}_n|^2L^{(k)}_n}{\sum\limits^{N}_{i=1, i\neq n}\zeta^{(k)}_iP^{(k)}|g^{(k)}_n|^2L^{(k)}_n + N_0}.
\label{eq:sinr_private}
\end{align}
\subsubsection{Case II: Full Frequency Reuse}
In the second scenario, all HAPS beams reuse the same frequency resource. Consequently, the received signal at $U_n^{(k)}$ consists of the desired signal from beam $k$, inter-beam interference from other beams, and AWGN. It is expressed as
\begin{align}
y_n^{(k)}=g_n^{(k)}\sqrt{L_n^{(k)}}x^{(k)}+\sum\limits_{\ell=1, \ell\neq k}^{K}g_n^{(\ell)}\sqrt{L_n^{(\ell)}}x^{(\ell)}+\eta_n^{(k)}.
\label{eq:rx_reuse}
\end{align}
where the second term represents the inter-beam interference. Note that, $L^{(\ell)}_n = 10^{F^{(l)}_n/10}$ such that $ F^{(\ell)}_n= G_T^{(k, \ell)} + G^{(k)}_{R,n} - L^{(\ell)}_{F,n} - L^{(\ell)}_{R,n} - L^{(\ell)}_A$. Here, $G_T^{(k, \ell)}$ is beam $\ell$'s transmit antenna gain toward $U_n^{(k)}$ actual location. Moreover, it is assumed that $g_n^{(\ell)}$ follows shadowed Rician distribution.
The common-stream SINR at user $U_n^{(k)}$ is therefore given by

\begin{equation}
\Upsilon_{c,n}^{(k)}=\frac{P^{(k)}\zeta_c^{(k)}|g_n^{(k)}|^2L_n^{(k)}
}{P^{(k)}(1-\zeta_c^{(k)})|g_n^{(k)}|^2L_n^{(k)}+I_{n}^{(k)}+N_0},
\label{eq:sinr_common_reuse}
\end{equation}
where
$I_n^{(k)}=\sum\limits_{\ell=1, \ell\neq k}^{K}P^{(\ell)}|g_n^{(\ell)}|^2
L_n^{(\ell)}$ denotes the inter-beam interference power.
After SIC of the common stream, the private-stream SINR can be expressed as

\begin{equation}
\Upsilon_{p,n}^{(k)}=\frac{P^{(k)}\zeta_n^{(k)}|g_n^{(k)}|^2L_n^{(k)}}{\sum\limits_{j\neq n}P^{(k)}\zeta_j^{(k)}|g_n^{(k)}|^2L_n^{(k)}+I_n^{(k)}+N_0}.
\label{eq:sinr_private_reuse}
\end{equation}

Equations \eqref{eq:sinr_common}, \eqref{eq:sinr_private}, 
\eqref{eq:sinr_common_reuse}, and \eqref{eq:sinr_private_reuse} 
fully characterize the common-stream and private-stream SINRs experienced by 
user $U_n^{(k)}$ for the two considered transmission scenarios. These SINR expressions serve as the basis for the performance analysis 
presented in the following section.

\section{Performance Analysis}\label{sec:analysis}
In this section, we derive the outage probability and throughput expressions of the considered HAPS-RSMA system.
The outage probability and throughput are evaluated separately for the two considered cases:
orthogonal frequency allocation and full frequency reuse. Let  $\gamma^{th, (k)}_{c,n}=2^{R^{(k)}_{c,n}}-1$ and $\gamma^{th, (k)}_{p,n}=2^{R^{(k)}_{p,n}}-1$ denote the target SINR thresholds for $U^{(k)}_n$ common and private stream, respectively with $R^{(k)}_{c,n}$ and $R^{(k)}_{c,n}$ representing the corresponding target rates. If \eqref{eq:sinr_common} and \eqref{eq:sinr_private} for \textit{Case I} and \eqref{eq:sinr_common_reuse} and \eqref{eq:sinr_private_reuse} for \textit{Case II} fails to cross the thresholds $\gamma^{th, (k)}_{c,n}$ and $\gamma^{th, (k)}_{p,n}$, respectively, then $U^{(k)}_n$ will be in an outage for \textit{Case I} and \textit{Case II}, respectively.

The respective outage probability for orthogonal frequency allocation (\textit{Case I}), is evaluated in the following proposition.
\begin{proposition}
    The outage probability for $U^{(k)}_n$ considering orthogonal frequency allocation can be expressed as
    \begin{align}\label{prop:outage_ortho}
P_{\mathrm{out},n}^{(k)} = \max\left(\mathcal{P}_{c,n}^{(k)},\ \mathcal{P}_{p,n}^{(k)}\right),
    \end{align}
 where    \begin{align}
 \mathcal{P}_{c,n}^{(k)}=
 1-
\sum_{i=0}^{m^{(k)}_n-1}\sum_{l=0}^{i}
\binom{m^{(k)}_n-1}{i}
\frac{\alpha^{(k)}_n(\delta^{(k)}_n)^i}{l!}\nonumber\\
\exp\left({\frac{-A_1\gamma_{c,n}^{th, (k)}}{A_2}}\right)\frac{\left(\frac{\gamma_{c,n}^{th, (k)}}{A_2}\right)^l}{A_1^{i-l+1}},
\nonumber  \end{align}
\begin{align}
 \mathcal{P}_{p,n}^{(k)}=
 1-
\sum_{i=0}^{m^{(k)}_n-1}\sum_{l=0}^{i}
\binom{m^{(k)}_n-1}{i}
\frac{\alpha^{(k)}_n(\delta^{(k)}_n)^i}{l!}\nonumber\\
\exp\left({\frac{-A_1\gamma_{p,n}^{th, (k)}}{A_3}}\right)\frac{\left(\frac{\gamma_{p,n}^{th, (k)}}{A_3}\right)^l}{A_1^{i-l+1}},
\nonumber  \end{align}
provided $A_2>0$ and $A_3>0$ respectively; otherwise $\mathcal{P}_{c,n}^{(k)}=1$ and $\mathcal{P}_{p,n}^{(k)}=1$, such that $\rho^{(k)}=P^{(k)}/N_{0}$,   $A_3=\rho^{(k)}L_n^{(k)}\left(\zeta_n^{(k)}-\gamma_{p,n}^{th, (k)}\sum^{N}_{i=1, i\neq n} \zeta^{(k)}_i\right)$,  $A_1=(\beta^{(k)}_n-\delta^{(k)}_n)$, and  $A_2=\rho^{(k)}L_n^{(k)}\left(\zeta_c^{(k)}-\gamma_{c,n}^{th, (k)}\left\{1-\zeta_c^{(k)}\right\}\right)$.
\end{proposition}
\begin{proof}
    Please refer to Appendix~A.
\end{proof}
\begin{remark}
    The throughput for $U^{(k)}_n$ considering orthogonal frequency allocation can be expressed as
    \begin{align}
        T^{(k)}_{n}\label{throu_1}
= \left(1-P_{\mathrm{out},n}^{(k)} \right)\left({R^{(k)}_{c,n}}+{R^{(k)}_{p,n}}\right).   \end{align}
\end{remark}
Next, we focus on deriving the expression for \textit{Case II} (full frequency reuse). As stated earlier the common-stream SINR of user $U_n^{(k)}$ is given by \eqref{eq:sinr_common_reuse}, where the aggregate inter-beam interference power is
$I_n^{(k)} = \sum_{\ell=1,\ell\neq k}^{K} P^{(\ell)}|g_n^{(\ell)}|^2L_n^{(\ell)}$,
and each $g_n^{(\ell)}$ is independently shadowed-Rician distributed with parameters $\left(m_n^{(\ell)},b_n^{(\ell)},\Omega_n^{(\ell)}\right)$. Because $I_n^{(k)}$ is itself random, the resulting decoding threshold on $Z_n^{(k)}=|g_n^{(k)}|^2$ is no longer fixed, and the outage probability must be obtained by averaging over the distribution of $I_n^{(k)}$. We derive this in closed form by first approximating $I_n^{(k)}$ as a Gamma random variable via moment matching, and then exploiting the Laplace transform of the Gamma distribution to complete the averaging in closed form.

\subsection{Moments of a Single Interferer}
\begin{lemma}
\label{lemma:moments}
For a shadowed-Rician power gain $Z_n^{(\ell)}=|g_n^{(\ell)}|^2$ with parameters $\left(m_n^{(\ell)},b_n^{(\ell)},\Omega_n^{(\ell)}\right)$,
\begin{align}
\mathbb{E}\!\left[Z_n^{(\ell)}\right] &= \Omega_n^{(\ell)}+2b_n^{(\ell)}, \nonumber\\
\mathrm{Var}\!\left[Z_n^{(\ell)}\right] &= 4\left(b_n^{(\ell)}\right)^2+4b_n^{(\ell)}\Omega_n^{(\ell)}+\frac{\left(\Omega_n^{(\ell)}\right)^2}{m_n^{(\ell)}}.
\label{eq:moments_single}
\end{align}
\end{lemma}
\begin{proof}
    Please refer to Appendix~B.
\end{proof}
\subsection{Gamma Approximation of the Aggregate Interference}
Since $I_n^{(k)}$ is a sum of independent, scaled versions of $Z_n^{(\ell)}$, its mean and variance follow directly from Lemma~\ref{lemma:moments}:
$\kappa_I^{(k)} = \frac{\left(\mu_I^{(k)}\right)^2}{\left(\sigma_I^{(k)}\right)^2}, \left(\sigma_I^{(k)}\right)^2 = \sum_{\ell\neq k}\left(P^{(\ell)}L_n^{(\ell)}\right)^{\!2}
\left(4\left(b_n^{(\ell)}\right)^2+4b_n^{(\ell)}\Omega_n^{(\ell)}+\frac{\left(\Omega_n^{(\ell)}\right)^2}{m_n^{(\ell)}}\right)
\Theta_I^{(k)} = \frac{\left(\sigma_I^{(k)}\right)^2}{\mu_I^{(k)}}, 
\mu_I^{(k)} = \sum_{\ell\neq k}P^{(\ell)}L_n^{(\ell)}\left(\Omega_n^{(\ell)}+2b_n^{(\ell)}\right)$.
We approximate $I_n^{(k)}$ as Gamma distributed, $I_n^{(k)}\sim\mathrm{Gamma}\!\left(\kappa_I^{(k)},\Theta_I^{(k)}\right)$, with shape and scale parameters obtained by matching the first two moments:
\begin{equation}
\kappa_I^{(k)} = \frac{\left(\mu_I^{(k)}\right)^2}{\left(\sigma_I^{(k)}\right)^2},\;\; \Theta_I^{(k)} = \frac{\left(\sigma_I^{(k)}\right)^2}{\mu_I^{(k)}}.
\label{eq:gamma_params}
\end{equation}
The respective outage probability for full frequency reuse (\textit{Case II}), is evaluated in the following proposition.
\begin{proposition}
    The outage probability for $U^{(k)}_n$ considering full frequency reuse can be evaluated as
    \begin{align}\label{outage_full_freq}
\hat{P}_{\mathrm{out},n}^{(k)} =\max\left(\hat{\mathcal{P}}_{c,n}^{(k)},\ \hat{\mathcal{P}}_{p,n}^{(k)}\right),
    \end{align}
 where \begin{align}
\hat{\mathcal{P}}_{c,n}^{(k)} \approx
1-\sum_{i=0}^{m_n^{(k)}-1}\sum_{l=0}^{i}\sum_{\tilde{p}=0}^{l}\binom{m_n^{(k)}-1}{i}\binom{l}{\tilde{p}}\left(\delta_n^{(k)}\right)^{i}\alpha_n^{(k)}\nonumber\\
\times\exp\left({-V_1 N_0}\right)
\frac{t_c^{\,l}}{l!\left(A_1\right)^{i-l+1}}
 (N_0)^{l-\tilde{p}}\!\left(\Theta_I^{(k)}\right)^{\tilde{p}}\!\left(\kappa_I^{(k)}\right)_{\tilde{p}}\nonumber\\
\times\left(1+\Theta_I^{(k)}V_1\right)^{-\left(\kappa_I^{(k)}+\tilde{p}\right)}\nonumber.
\end{align}
\begin{align}
\hat{\mathcal{P}}_{p,n}^{(k)} \approx
1-\sum_{i=0}^{m_n^{(k)}-1}\sum_{l=0}^{i}\sum_{\tilde{p}=0}^{l}\binom{m_n^{(k)}-1}{i}\binom{l}{\tilde{p}}\left(\delta_n^{(k)}\right)^{i} \alpha_n^{(k)}\nonumber\\
\times \exp\left(-V_2 N_0\right)
\frac{t_p^{\,l}}{l!\left(A_1\right)^{i-l+1}} (N_0)^{l-\tilde{p}}\!\left(\Theta_I^{(k)}\right)^{\tilde{p}}\!\left(\kappa_I^{(k)}\right)_{\tilde{p}}\nonumber\\
\times\left(1+\Theta_I^{(k)}V_2\right)^{-\left(\kappa_I^{(k)}+\tilde{p}\right)}\nonumber.
\end{align}
where $\tilde{A}_2=P^{(k)}L_n^{(k)}\left(\zeta_c^{(k)}-\gamma_{c,n}^{th, (k)}\left\{1-\zeta_c^{(k)}\right\}\right)$, $\tilde{A}_3=P^{(k)}L_n^{(k)}\left(\zeta_n^{(k)}-\gamma_{p,n}^{th, (k)}\sum^{N}_{i=1, i\neq n} \zeta^{(k)}_i\right)$ $V_1=A_1 t_c$, $V_2=A_1 t_p$, $t_c=\frac{\gamma_{c,n}^{th,(k)}}{\tilde{A}_2}$, $t_p=\frac{\gamma_{p,n}^{th,(k)}}{\tilde{A}_3}$,   and $\left(\kappa_I^{(k)}\right)_{\tilde{p}}\triangleq \kappa_I^{(k)}\left(\kappa_I^{(k)}+1\right)\cdots\left(\kappa_I^{(k)}+{\tilde{p}}-1\right)$ is the rising factorial. Both expressions hold provided $\tilde{A}_2>0$ and $\tilde{A}_3>0$, respectively; otherwise the outage probability is $1$.
\end{proposition}
\begin{proof}
    Please refer to Appendix~C.
\end{proof}
\begin{remark}
The throughput of user $U_n^{(k)}$ under full frequency reuse can be obtained by replacing ${P}_{\mathrm{out},n}^{(k)}$ with $\hat{P}_{\mathrm{out},n}^{(k)}$ in \eqref{throu_1}.
\end{remark}

\section{Numerical Results}
\begin{figure*}[t]
    \centering
   \begin{minipage}[b]{0.32\textwidth}
        \centering
       \includegraphics[width=\textwidth]{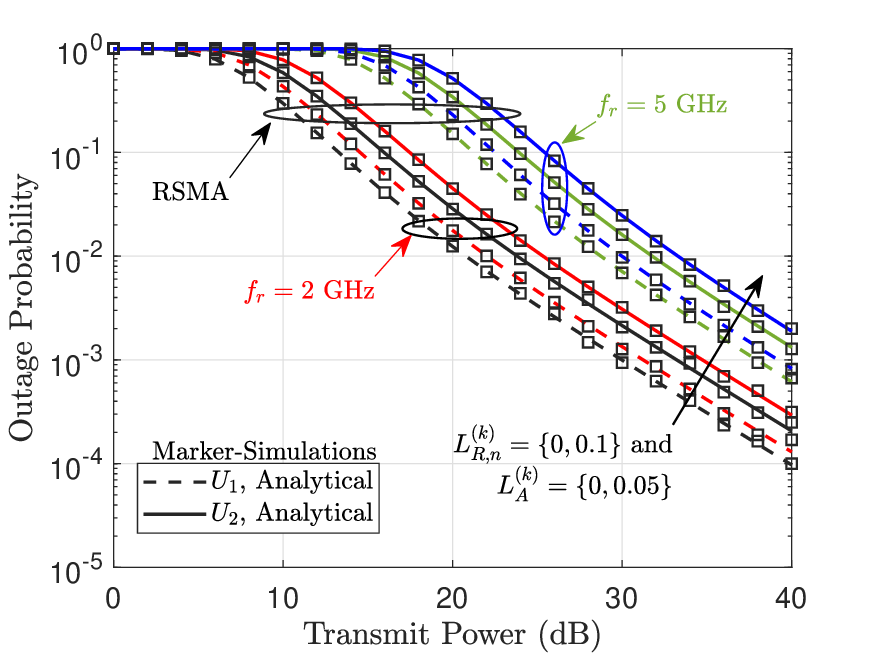}
       \caption{\textit{Case I}: Outage probability.}
        \label{fig1}
    \end{minipage}
     \begin{minipage}[b]{0.32\textwidth}
        \centering
       \includegraphics[width=\textwidth]{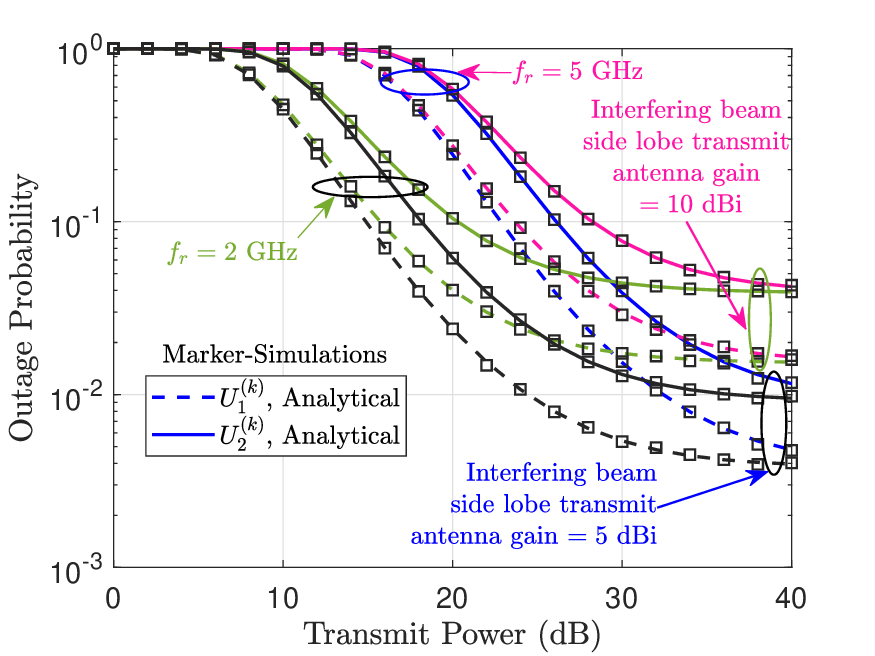}
        \caption{\textit{Case II}: Outage probability.}
        \label{fig2}
    \end{minipage}
    \begin{minipage}[b]{0.32\textwidth}
        \centering
        \includegraphics[width=\textwidth]{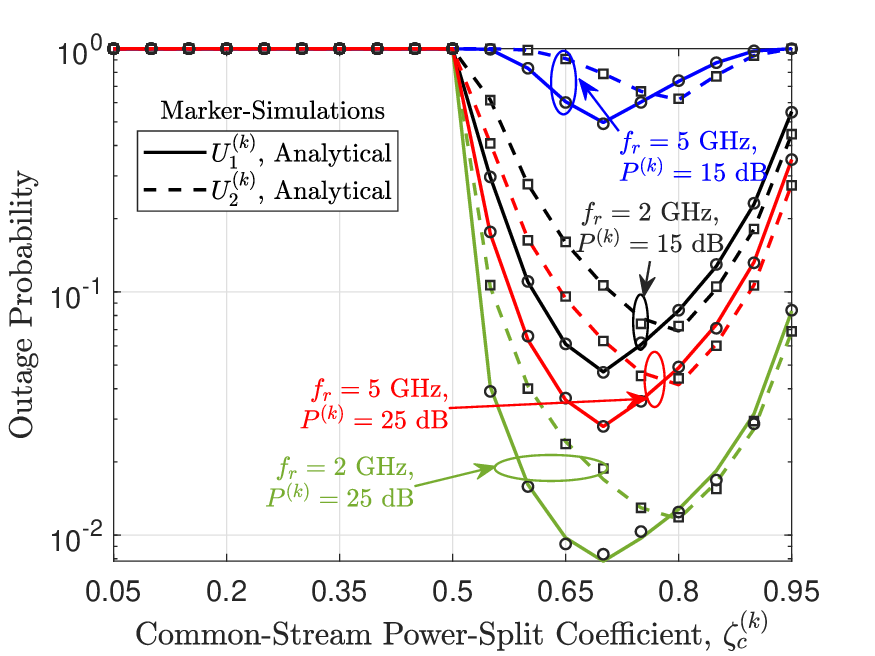}
        \caption{\textit{Case II}: Outage probability vs $\zeta_c^{(k)}$.}
        \label{fig3}
    \end{minipage}
    
\end{figure*}
This section validates the analytical outage expressions derived in Section~\ref{sec:analysis} through Monte Carlo simulations and examines the outage performance under orthogonal frequency allocation and full frequency reuse. Note that all the derived expressions are valid for $\mathcal{N}$ users in both \textit{Case I} and \textit{Case II}. However, for clarity, the results are presented only for two downlink users, $U^{(k)}_1$ and $U^{(k)}_2$, in both cases.

Unless otherwise specified, the complete set of simulation parameters is provided in Table~\ref{sim_table}. Additionally, achievable throughput follows directly from the derived outage expressions as stated in Remark~1 and Remark~2. Therefore, due to space limitations, only outage results are presented.

\begin{table}[t]	\renewcommand{\arraystretch}{1}
		\centering
		\caption{ Simulation Parameters.}
		\label{sim_table}
			\resizebox{\columnwidth}{!}{\begin{tabular}{|l|l|l|l|l|l|}
			\hline
			Parameter         & Value         & Parameter & Value  & Parameter & Value  \\ \hline
			$H  $   &     $20 $~km   	&  $\theta_1^{(k)}  $   &     $60^{\circ}$ &  $\theta_2^{(k)}  $   &     $45^{\circ}$ \\ \hline
			
			$G_{T}^{(k)}$    &    $ 30$~dBi &	  $G^{(k)}_{R,n}$ &     $3$~dBi &   $R_{c,n}^{(k)}$ &     $1$  	 \\ \hline	

           $R_{p,1}^{(k)} $   &     $0.6 $ &	 $R_{p,2}^{(k)} $&     $0.1$ &   $m_1^{(k)}$ &     $10$  	 \\ \hline

           $b_1^{(k)}$   &     $0.158 $ &	  $\Omega_1^{(k)}$ &     $1.29$ &  $m_2^{(k)}$ &     $8$ 	 \\ \hline	

$b_2^{(k)}$   &     $0.140 $ &	  $\Omega_2^{(k)}$ &     $1.10$ &  $N_0$ &     -$90$~dB 	 \\ \hline
		\end{tabular}}
	\end{table}  


Fig.~\ref{fig1} and Fig.~\ref{fig2} plot outage probability versus the $k_{th}$-beam transmit power $P^{(k)}$ for orthogonal allocation and full frequency reuse, respectively, across carrier frequencies $f_r$ and link-impairment values $L_{R,n}^{(k)}$, $L_A^{(k)}$. For the full frequency reuse results, the $k_{th}$ beam is assumed to experience interference from two co-channel beams ($K-1=2$) located at $\theta^{(\ell)}={50^\circ,30^\circ}$, with shadowed-Rician parameters ${(4,0.4,0.6),(3,0.5,0.4)}$ and sidelobe transmit gains of $5$ and $10$~dBi. In both, the analytical curves match simulation closely, validating Propositions~1 and~2. It can be observed that increasing $P^{(k)}$ improves outage in all cases, while higher $f_r$, $L_{R,n}^{(k)}$, and $L_A^{(k)}$ each degrade it by reducing the net link gain $F_n^{(k)}$ in \eqref{eq:link_budget}. Consequently, a higher transmit power is required to achieve the same outage performance, resulting in a rightward shift of the outage curves.

A key distinction appears at high transmit power: under orthogonal allocation, the outage probability continues to decrease with increasing $P^{(k)}$, over the considered power range, as no additional co-channel interference from other beams is present. In contrast, under full frequency reuse it saturates to a non-zero floor, since the desired signal and aggregate interference $I_n^{(k)}$ scale together with the transmit power (i.e., $P^{(k)}=P^{(\ell)}\;\forall\ell$), driving the SINR to a power-independent desired-to-interference ratio. This floor is nearly identical for $2$~GHz and $5$~GHz for both $5$~dBi and $10$~dBi sidelobe transmit antenna gain from beams other than the $k_{th}$ beam.  Since the desired and interfering links share the same carrier frequency and path-loss formula, the frequency-dependent term cancels in their ratio, leaving the floor set by the mainlobe-to-sidelobe gain difference rather than $f_r$.
Carrier frequency instead determines the power needed to \emph{reach} the floor, with $5$~GHz remaining noise-limited over a wider power range before converging to the same floor as $2$~GHz. This confirms that transmit power alone cannot overcome full-frequency-reuse interference, motivating sidelobe suppression or power-split design as an important way forward for better system performance.

Fig.~\ref{fig3} plots outage versus the common-stream power-split coefficient $\zeta_c^{(k)}$ under full frequency reuse, for $2$~GHz and $5$~GHz with transmit power set at $15$ and $25$~dB, with $L_{R,n}^{(k)}=0.1$, $L_A^{(k)}=0.05$, and $10$~dBi sidelobe gain. Outage is identically one for $\zeta_c^{(k)}\le0.5$, a threshold obtained directly from the feasibility condition $\tilde{A}_2>0$: with $\gamma_{c,n}^{th,(k)}=2^{R_{c,n}^{(k)}}-1=1$ (as ${R_{c,n}^{(k)}}$ is set to $1$), feasibility requires $\zeta_c^{(k)}>\gamma_{c,n}^{th,(k)}/(1+\gamma_{c,n}^{th,(k)})=0.5$, below which the common-stream rate is unsupportable.

Beyond this cutoff, outage decreases as the common-stream margin widens, then rises again as the shrinking private-stream budget $(1-\zeta_c^{(k)})$ increasingly limits $\mathcal{P}_{p,n}^{(k)}$, so that the overall outage $\max(\mathcal{P}_{c,n}^{(k)},\mathcal{P}_{p,n}^{(k)})$ follows a non-monotonic trend governed by the balance between these two competing terms. This balance is user-dependent: $U_2$'s channel condition makes $\mathcal{P}_{c,2}^{(k)}$ fall more slowly past the cutoff, while its smaller private power share makes $\mathcal{P}_{p,2}^{(k)}$ rise more sharply as $\zeta_c^{(k)}$ increases, so $U_2$ requires a noticeably larger $\zeta_c^{(k)}$ than $U_1$ to reach comparably low outage. This asymmetry shows that a single fixed $\zeta_c^{(k)}$ cannot equally favor both users, motivating fairness-aware or adaptive power-split design as a natural extension of this analysis.

\section{Conclusion}
This paper investigated the outage performance of HAPS-assisted downlink RSMA under realistic link-budget conditions and shadowed-Rician fading. Closed-form outage expressions were derived for orthogonal frequency allocation, while a tractable Gamma approximation was developed for full frequency reuse to characterize the impact of aggregate inter-beam interference. The analytical results were validated through Monte Carlo simulations, confirming the accuracy of the proposed expressions. Numerical results showed that aggregate inter-beam interference in full frequency reuse leads to an outage floor, whereas orthogonal frequency allocation achieves continuously improving outage performance with increasing transmit power. These findings provide useful insights for the design of power allocation and frequency reuse strategies in HAPS-assisted RSMA networks.
\appendices
\section{}
\begin{proof}
From \eqref{eq:sinr_common} and \eqref{eq:sinr_private}, both $\gamma_{c,n}^{(k)}$ and $\gamma_{p,n}^{(k)}$ are monotonically increasing functions of the same random variable $Z_n^{(k)}=|g_n^{(k)}|^2$. Rearranging $\gamma_{c,n}^{(k)}\geq\gamma_{c,n}^{th,(k)}$ and $\gamma_{p,n}^{(k)}\geq\gamma_{p,n}^{th,(k)}$ yields the equivalent conditions $Z_n^{(k)}\geq z_{c,n}^{(k)}\triangleq\gamma_{c,n}^{th,(k)}/A_2$ and $Z_n^{(k)}\geq z_{p,n}^{(k)}\triangleq\gamma_{p,n}^{th,(k)}/A_3$, respectively (with outage occurring with probability one if $A_2\leq0$ or $A_3\leq0$). Since $U_n^{(k)}$ is outage-free only if \emph{both} conditions hold, and both are lower-bound constraints on the same $Z_n^{(k)}$, their intersection reduces to a single condition at the larger threshold:
$\left\{Z_n^{(k)}\geq z_{c,n}^{(k)}\right\}\cap\left\{Z_n^{(k)}\geq z_{p,n}^{(k)}\right\} \!=\! \left\{\!Z_n^{(k)}\!\geq\max\left(\!z_{c,n}^{(k)},z_{p,n}^{(k)}\right)\!\right\}$.
The outage probability thus, follows from the CDF of $Z_n^{(k)}$:
\begin{align}
&P_{\mathrm{out},n}^{(k)} = F_{Z_n^{(k)}}\!\left(\max\left(z_{c,n}^{(k)},z_{p,n}^{(k)}\right)\right) \nonumber\\= &\max\!\left(\!F_{Z_n^{(k)}}\!\left(\!z_{c,n}^{(k)}\!\right),\!F_{Z_n^{(k)}}\!\left(\!z_{p,n}^{(k)}\!\right)\!\right) \!=\! \max\!\left(\mathcal{P}_{c,n}^{(k)},\mathcal{P}_{p,n}^{(k)}\right),
\end{align}
where the second equality holds because $F_{Z_n^{(k)}}(\cdot)$ is monotonically non-decreasing.

\end{proof}
\section{}

\begin{proof}
The shadowed-Rician model represents $g_n^{(\ell)}=S+N$, where $S$ is a line-of-sight component with random power $\xi\sim\mathrm{Gamma}\!\left(m_n^{(\ell)},\Omega_n^{(\ell)}/m_n^{(\ell)}\right)$ and uniform random phase, and $N\sim\mathcal{CN}(0,2b_n^{(\ell)})$ is an independent scattered component. Conditioned on $\xi$, $Z_n^{(\ell)}=|S+N|^2$ follows a noncentral chi-squared distribution with two degrees of freedom, noncentrality parameter $\xi$, and per-dimension variance $b_n^{(\ell)}$, for which
\begin{equation}
\mathbb{E}\!\left[Z_n^{(\ell)}\mid\xi\right]\!=\!\xi+2b_n^{(\ell)},
\mathrm{Var}\!\left[Z_n^{(\ell)}\!\mid\!\xi\right]\!=\!4\left(b_n^{(\ell)}\right)^2\!\!+4b_n^{(\ell)}\xi.
\label{eq:conditional_moments}
\end{equation}
Taking the expectation of the first relation in \eqref{eq:conditional_moments} over $\xi$, and using $\mathbb{E}[\xi]=\Omega_n^{(\ell)}$, gives $\mathbb{E}[Z_n^{(\ell)}]=\Omega_n^{(\ell)}+2b_n^{(\ell)}$. By the law of total variance,
\begin{align}
\mathrm{Var}\!\left[Z_n^{(\ell)}\right] = \mathbb{E}\!\left[\mathrm{Var}\!\left[Z_n^{(\ell)}\mid\xi\right]\right] + \mathrm{Var}\!\left[\mathbb{E}\!\left[Z_n^{(\ell)}\mid\xi\right]\right] \nonumber\\= \left(4\left(b_n^{(\ell)}\right)^2+4b_n^{(\ell)}\Omega_n^{(\ell)}\right) + \mathrm{Var}[\xi].
\end{align}
Since $\xi\sim\mathrm{Gamma}\!\left(m_n^{(\ell)},\Omega_n^{(\ell)}/m_n^{(\ell)}\right)$, its variance is $\mathrm{Var}[\xi]=\left(\Omega_n^{(\ell)}\right)^2/m_n^{(\ell)}$, which yields \eqref{eq:moments_single}.
\end{proof}

\section{}
\begin{proof}
Rearranging $\Upsilon_{c,n}^{(k)}\geq\gamma_{c,n}^{th,(k)}$ in \eqref{eq:sinr_common_reuse} isolates $Z_n^{(k)}$, yielding the interference-dependent threshold condition $Z_n^{(k)}\geq t_c\!\left(I_n^{(k)}+N_0\right)$, provided $\tilde{A}_2>0$; the analogous step applied to $\Upsilon_{p,n}^{(k)}\geq\gamma_{p,n}^{th,(k)}$ gives $Z_n^{(k)}\geq t_p\left(I_n^{(k)}+N_0\right)$. Since $I_n^{(k)}$ is random, each outage probability is obtained by averaging the shadowed-Rician CDF of $Z_n^{(k)}$ over $I_n^{(k)}$ such that
$\hat{\mathcal{P}}_{c,n}^{(k)} = \mathbb{E}_{I_n^{(k)}}\!\left[F_{Z_n^{(k)}}\!\left(t_c\!\left(I_n^{(k)}+N_0\right)\right)\right]$
and similarly for $\mathcal{P}_{p,n}^{(k)}$ with $t_p$ in place of $t_c$. Substituting \eqref{eq:sr_cdf} and expanding $\left(I_n^{(k)}+N_0\right)^l$ via the binomial theorem\cite[ $1.111$]{2015249}, reduces the averaging to computing $\mathbb{E}\!\left[\left(I_n^{(k)}\right)^{\tilde{p}}\exp\left({-sI_n^{(k)}}\right)\right]$ for ${\tilde{p}}=0,\ldots,l$. Under the Gamma approximation of \eqref{eq:gamma_params}, $I_n^{(k)}$ admits the Laplace transform $\mathcal{L}_{I_n^{(k)}}(s)=\left(1+\Theta_I^{(k)}s\right)^{-\kappa_I^{(k)}}$, and the required moment follows from
\begin{align}
\mathbb{E}\!\left[\left(I_n^{(k)}\right)^{\tilde{p}}\exp\left({-sI_n^{(k)}}\right)\right] = (-1)^{\tilde{p}}\frac{d^{\tilde{p}}}{ds^{\tilde{p}}}\mathcal{L}_{I_n^{(k)}}(s)\nonumber\\ = \left(\Theta_I^{(k)}\right)^{\tilde{p}}\left(\kappa_I^{(k)}\right)_{\tilde{p}}
\left(1+\Theta_I^{(k)}s\right)^{-\left(\kappa_I^{(k)}+{\tilde{p}}\right)}.
\end{align}
Substituting this back through the binomial expansion and the outer CDF summation yields  $\hat{\mathcal{P}}_{c,n}^{(k)}$ (with $s=A_1 t_c$) and  $\hat{\mathcal{P}}_{p,n}^{(k)}$ (with $s=A_1 t_p$), completing the proof.
\end{proof}

\bibliographystyle{IEEEtran_renamed}
	\bibliography{refer}
\end{document}